\documentclass[lettersize,journal]{IEEEtran}
\usepackage{cite}
\usepackage{amsmath,amssymb,amsfonts}
\usepackage{amsthm}
\usepackage{algorithm}
\usepackage{graphicx}
\usepackage{textcomp}
\usepackage{xcolor}
\usepackage{bm}
\usepackage{balance}
\usepackage{makecell}
\usepackage{algpseudocode}
\usepackage{algorithm}
\usepackage{todonotes}
\usepackage{algorithm}
\usepackage{algpseudocode}
\usepackage{amsmath}

\floatstyle{boxed}

\newtheorem{theorem}{Theorem}
\newtheorem{lemma}{Lemma}

\newtheorem{definition}{Definition}
\newtheorem{remark}{Remark}
\newtheorem{example}{Example}
\def\BibTeX{{\rm B\kern-.05em{\sc i\kern-.025em b}\kern-.08em
    T\kern-.1667em\lower.7ex\hbox{E}\kern-.125emX}}
\begin{document}

\onecolumn

\title{Codes with Biochemical Constraints and Single Error Correction for DNA-Based Data Storage
}

\author{Shu Liu, Chaoping Xing, Yaqian Zhang
\thanks{
Shu Liu is with National Key Laboratory on Wireless Communications, University of Electronic Science and Technology of China, Chengdu, China (shuliu@uestc.edu.cn).
Chaoping Xing and Yaqian Zhang are with School of Electronic Information and Electric Engineering, Shanghai Jiao Tong University, Shanghai, China. (xingcp@sjtu.edu.cn, zhangyq9@sjtu.edu.cn).
}}

%

\maketitle

\begin{abstract}
In DNA-based data storage, DNA codes with biochemical constraints and error correction are designed to protect data reliability.
Single-stranded DNA sequences with secondary structure avoidance (SSA) help to avoid undesirable secondary structures which may cause chemical inactivity. Homopolymer run-length limit and GC-balanced limit also help to reduce the error probability of DNA sequences during synthesizing and sequencing. In this letter, based on a recent work \cite{bib7}, we construct DNA codes free of secondary structures of stem length $\geq m$ and have homopolymer run-length $\leq\ell$ for odd $m\leq11$ and $\ell\geq3$ with rate $1+\log_2\rho_m-3/(2^{\ell-1}+\ell+1)$, where $\rho_m$ is in Table \ref{tm}. In particular, when $m=3$, $\ell=4$, its rate tends to 1.3206 bits/nt, beating a previous work by  Benerjee {\it et al.}. We also construct DNA codes with all of the above three constraints as well as single error correction. At last, codes with GC-locally balanced constraint are presented.
\end{abstract}

\section{Introduction}
In DNA-based storage systems, data is translated to DNA sequences consisting of four nucleotides A, T, C, G, and stored in a DNA pool.
To protect data from errors, DNA sequences are desired to be able to correct insertion, deletion and substitution errors, and also satisfy some constraints that help to reduce error probability  \cite{b2, bib1}.

For DNA alphabet $\{{\rm A, T, C, G}\}$, the {\it Watson-Crick complement} gives $\overline{\rm A}={\rm T}$, $\overline{\rm T}={\rm A}$, $\overline{\rm C}={\rm G}$, $\overline{\rm G}={\rm C}$. The reverse-complement of a DNA sequence $x_1x_2...x_n$ with $x_i\in\{{\rm A, T, C, G}\}$ is defined as $\overline{x_n\cdots x_2x_1}$.
A secondary structure of a DNA sequence is formed by folding the sequence back upon itself, and this may lead to chemical inactivity of the sequence.
For example, in ATTCGGAA, the two subsequences TTC, GAA are reverse-complements of each other and may bind to each other after pairing of A with T and G with C. Thus it forms a secondary structure with a loop and a stem of length 3. For a positive integer $m$, an $m$-SSA (secondary structure avoidance) sequence does not contain two non-overlapping reverse-complement consecutive subsequences of length $m$. The SSA constraint can reduce error probability of DNA sequences when reading data \cite{bib1}. Another constraint is homopolymer run-length limit, indicating that the maximum number of repetitive consecutive symbols in a DNA sequence should be $\leq\ell$, which is called $\ell$-run-length limited property.
Moreover, the percentage of nucleotides G and C in the whole sequence (GC content) is expected to be about $50\%$, called GC-balanced limit constraint.

For a quaternary DNA code $\mathcal{C}$ of length $n$, the code rate is defined as $\log_2{|\mathcal{C}|}/n$ (bits/nt), which indicates the storage overhead.
A challenging problem is to construct DNA codes satisfying multiple constraints and error correction with high code rate.
In the literature, several works \cite{BB2021, bib6, bib7} are presented to construct $m$-SSA codes for different $m$. In particular, \cite{bib7} investigated codes for $m\leq6$, which achieve optimal rate for $m=2, 3$. The work \cite{BB2021} designed 3-SSA and 4-run-length limited codes with rate 1.1609.
Run-length limited and GC-balanced codes are constructed in \cite{b9, bib8, bib9, bib10}.
The work \cite{TCS} considered 3-SSA with run-length limit and GC-balanced limit using constacyclic codes. \cite{Benerjee} gave codes that satisfy 3-SSA, 3-run-length limited and GC-balanced properties as well as reverse constraint and reverse-complement constraint, with code rate 0.8617.
The works \cite{b10, b11} constructed  $\ell$-run-length limited and GC-balanced codes with GC content within the range $0.5\pm\epsilon$, called $(\epsilon, \ell)$-constrained codes. The work \cite{bib11capacity} studied capacity-achieving $(\epsilon, \ell)$-constrained codes.

In the scenario of error correction, we call insertion errors and deletion errors are insertion and deletion (insdel for short) errors, and edit errors present insertion, deletion and substitution errors.
\cite{b10, b11} extended the constrained codes to also correct single insdel/edit error.
Cai {\it et al.} \cite{bib12local} partitioned each codeword to local segments (partitions) of same length $s$ and gave a code design to enable error correction and almost GC balance in each segment with code rate about $2(1-\log_2s/s)$ where $s=\Omega(\log n)$.
Another definition of locally balanced constraint is introduced in \cite{b12, 2022ISIT} and means that any length-$s$ subsequence of a DNA codeword of length $n$ satisfies the GC-balanced property (called GC-locally balanced).

In this letter, we consider the problem of constructing DNA codes with multiple constraints and error correction. The main contributions are as follows.
\begin{itemize}
\item We present a construction of $m$-SSA and $\ell$-run-length limited DNA codes for odd $m\leq11$ and $\ell\geq3$. It has code rate $1+\log_2\rho_m-3/(2^{\ell-1}+\ell+1)$, where $\rho_m$ is given in Table \ref{tm}. In particular, when $m=3$ and $\ell=4$, the code has rate approximately 1.3206, which performs better than the code in \cite{BB2021}.
\item For any integer $\ell\geq3$, and a small real $0<\epsilon<0.5$, we construct 3-SSA, $\ell$-run-length limited and GC-balanced codes with GC content within $0.5\pm\epsilon$. The codes are further extended to also correct single edit error by a concatenation technique using Hamming codes and quaternary VT codes with rate approximately $(1-\frac{\log t}{t})(1+\log_2\rho)$, where $t$ is a positive integer and $\rho$ is the largest real root of the equation $x^{\ell+1}-\sum_{i=0}^{\ell-2}x^{i}=0$.
\item We give a transformation to convert GC-balanced codes to GC-locally balanced codes. The obtained GC-locally balanced codes also maintain 3-SSA, $\ell$-run-length limited and GC-balanced property with GC content within $0.5\pm\epsilon$, for some $\ell\geq3$, and $0<\epsilon<0.5$, as well as error correction property.
\end{itemize}

The rest of the letter is organized as follows. Section \ref{sec_pre} introduces some preliminaries. Section \ref{sec-3} presents $m$-SSA and $\ell$-run-length limited DNA codes for any odd $m\leq11$ and $\ell\geq3$. Section \ref{sec-4} gives DNA codes with multiple constraints as well as error correction. GC-locally balanced codes are constructed in Section \ref{sec-5}.

\section{Preliminary}\label{sec_pre}
\subsection{Notations and definitions}
For positive integers $n$ and $i<j$, denote $[n]=\{1,...,n\}$, $[i,j]=\{i,i+1,...,j\}$ and $[i,j)=\{i,i+1,...,j-1\}$. Let $\mathbb{F}_q$ be a finite field and $\mathbb{Z}_q=\{0,1,...,q-1\}$. Denote $\sum_{\rm DNA}=\{{\rm A, T, C, G}\}$ to be the DNA alphabet. Define the complement $\overline{\rm A}={\rm T}$, $\overline{\rm T}={\rm A}$, $\overline{\rm C}={\rm G}$, $\overline{\rm G}={\rm C}$.
Let $\bm{x}=(x_1,\cdots,x_n)\in\sum_{\rm DNA}^n$, we say a subsequence of $\bm{x}$ when it contains consecutive symbols of $\bm{x}$.
The {\it reverse-complement} of $\bm{x}$ is defined as ${\rm RC}(\bm{x})=\overline{x_n,\cdots,x_1}$.
For some positive integer $m$, we say $\bm{x}$ is {\it $m$-SSA} (secondary structure avoidance), if $\bm{x}$ does not have two non-overlapping subsequences $\bm{y}$ and $\bm{y}'$ of length $\geq m$ such that $\bm{y}={\rm RC}(\bm{y}')$.
Define the {\it GC weight} of $\bm{x}$ to be $\mathrm{wt}_{\mathrm{GC}}(\bm{x})=|\{i\in[n]: x_i={\rm G}~{\rm or} ~{\rm C}\}|$.

\begin{definition}
Let $0<\epsilon<0.5$ be a small real and $n, s$ be two positive integers.
Let $\bm{x}=(x_1, \cdots, x_n)\in\sum_{\rm DNA}^n$.
\begin{itemize}
\item \textbf{GC-$\epsilon$-globally balanced.} $\bm{x}$ is called GC-$\epsilon$-globally balanced if $|\frac{\mathrm{wt}_{\mathrm{GC}}(\bm{x})}{n}-0.5|\leq\epsilon$.
\item  \textbf{GC-$(s,\epsilon)$-partition balanced.} Suppose $s\mid n$ and $\bm{x}$ is partitioned into $\frac{n}{s}$ segments of size $s$, that is, $\bm{x}=(\bm{x}^{(1)},\cdots,\bm{x}^{(\frac{n}{s})})$, where $\bm{x}^{(i)}=(x _{(i-1)s+1},\cdots,x_{is})$, $i\in[\frac{n}{s}]$. Then $\bm{x}$ is called {\it GC-$(s,\epsilon)$-partition balanced} if $|\frac{\mathrm{wt}_{\mathrm{GC}}(\bm{x}^{(i)})}{s}-0.5|\leq\epsilon$ for all $i\in[\frac{n}{s}]$.
\item \textbf{GC-$(s,\epsilon)$-locally balanced.} Suppose $s<n$ and for $1\leq i\leq n-s+1$, denote $\bm{x}_i=(x_{i},x_{i+1},\cdots,x_{i+s-1})$, then $\bm{x}$ is called {\it GC-$(s,\epsilon)$-locally balanced} if $|\frac{\mathrm{wt}_{\mathrm{GC}}(\bm{x}_i)}{s}-0.5|\leq\epsilon$ for all $i\in[n-s+1]$.
\end{itemize}
\end{definition}

It is obvious that when $s\mid n$, a GC-$(s,\epsilon)$-locally balanced sequence is also GC-$(s,\epsilon)$-partition balanced and GC-$\epsilon$-globally balanced.
For a binary sequence $\bm{z}\in\mathbb{Z}_2^n$, similar definitions can be obtained by replacing GC-weight by Hamming weight of $\bm{z}$ and removing the term ``GC".
In the following, for GC-$\epsilon$-globally balanced sequences, we usually omit ``globally", when it is clear from the context.

\begin{definition}
Let $\bm{x}\in\sum_{\rm DNA}^n$ (resp. $\mathbb{Z}_2^n$) and write $\bm{x}=\left(a_1^{(n_1)},\cdots,a_t^{(n_t)}\right)$ for some positive integers $n_1,\cdots,n_t$ and $a_1,\cdots,a_t\in\sum_{\rm DNA}$ (resp. $\mathbb{Z}_2$) satisfying $n_1+\cdots+n_t=n$ and $a_i\neq a_{i+1}$ for all $1\leq i\leq t-1$. That is, $a_i^{(n_i)}$ contains $n_i$ consecutive symbols of $a_i$. We call each $a_i^{(n_i)}$ a run of length $n_i$ for $i\in[t]$.
\end{definition}

A code is said to be {\it $\ell$-run-length limited} if every codeword has run length at most $\ell$ for some integer $\ell>0$.

\subsection{Varshamov-Tenengolts (VT) codes}
The {\it Varshamov-Tenengolts codes} are designed for the insertion/deletion channel and are asymptotically optimal single insdel correcting codes.
In \cite{b15}, Tenengolts generalized binary VT codes to $q$-ary ones for $q>2$.
Next we review the definition of $q$-ary VT codes which will be used later.
For a $q$-ary sequence $\bm{x}=(x_1,\cdots,x_n)\in \mathbb{Z}_{q}^n$, define the {\it signature} of $\bm{x}$ to be the binary vector $\alpha(\bm{x})=(\alpha_1,\cdots,\alpha_{n-1})\in\{0,1\}^{n-1}$, where $\alpha_i=1$ if $x_{i+1}\geq x_{i}$ and $\alpha_i=0$ if $x_{i+1}<x_{i}$, for $i\in[n-1]$.
Define the {\it binary VT syndrome} of $\alpha(\bm{x})$ to be $\mathrm{Syn}(\alpha(\bm{x}))=\sum_{i=1}^{n-1}i\alpha_i$.

\begin{definition}[$q$-ary VT codes]
For $a\in\mathbb{Z}_{n}$ and $b\in\mathbb{Z}_q$, let
\begin{equation*}
\begin{aligned}
\mathrm{T}_{a,b}(n;q)=\{\bm{x}\in\mathbb{Z}_q^n:\ &\mathrm{Syn}(\alpha(\bm{x}))\equiv a \ (\mathrm{mod}\ n)   \\
 &\mathrm{and} \ \sum_{i=1}^{n}x_i\equiv b\ (\mathrm{mod}\ q)\},
\end{aligned}
\end{equation*}
then $\mathrm{T}_{a,b}(n;q)$ gives a class of $q$-ary VT codes.
\end{definition}

The $q$-ary VT codes can correct one insdel error with a linear-time decoder. For any $n$ and $q$, there exist some $a\in\mathbb{Z}_{n}$ and $b\in\mathbb{Z}_q$ s.t. $|\mathrm{T}_{a,b}(n;q)|\geq \frac{q^n}{qn}$ by the pigeonhole principle. Tenengolts provided a systematic encoder with redundancy $\log_2{n}+C_q$, where $C_q$ is independent of the code length $n$.

\section{Code constructions with secondary structure avoidance and run-length limit}\label{sec-3}
Given some $m>0$, we adopt the notion of {\it TC-m-dominant} in \cite{bib7} as follows.
A DNA sequence $\bm{x}$ is called TC-$m$-dominant, if for every subsequence of length $m$, the sum of appearances of T and C is larger than $m/2$.
Similarly, a binary sequence is called ``0"-$m$-dominant, if for every subsequence of length $m$, the sum of appearances of the symbol ``0" is larger than $m/2$.
Let $S_{\rm TC}(m,n)\subseteq\sum_{\rm DNA}^n$ be the set of TC-$m$-dominant DNA sequences of length $n$ and $S_0(m,n)\subseteq\mathbb{Z}_2^n$ be the set of ``0"-$m$-dominant binary sequences of length $n$. \cite{bib7} has investigated $S_{\rm TC}(m,n)$ and $S_0(m,n)$ for odd $m\leq11$ using recursive construction as well as a brute-force algorithm.
It is also illustrated in \cite{bib7} that a TC-$m$-dominant sequence must be an $m$-SSA sequence.
\begin{lemma}\label{TC-dominant}\cite{bib7}
Wnen $m$ is odd, a TC-$m$-dominant sequence must be an $m$-SSA sequence. Thus a code $\mathcal{C}$ consisting of all TC-$m$-dominant sequences of length $n$ is an $m$-SSA code.
\end{lemma}

\begin{lemma}\label{TC-3}\cite{bib7}
For odd $m\leq11$, $|S_0(m,n)|$ is approximately $\rho_m^n$, where $\rho_m$ is given in Table \ref{tm}.
\end{lemma}

{\footnotesize
\begin{table}[h]
  \centering
\caption{The Value of $\rho_m$}\label{tm}
\begin{tabular}{c|c|c|c|c|c}
\hline
        $m$    &3 &5 &7 &9 &11 \\ \hline
$\rho_m$ &1.5515 &1.6980 &1.7698 & 1.8131 &1.8423 \\ \hline
\end{tabular}
\end{table}
}

In the following, we present $m$-SSA and $\ell$-run-length limited codes for odd $m\leq11$ and $\ell\geq3$ using $S_0(m,n)$ in \cite{bib7}.
Define a map $\tau$ from DNA alphabet to binary pairs as follows.
\begin{equation}\label{correspondence}
\tau({\rm T})=00, ~~\tau({\rm C})=01, ~~\tau({\rm A})=10, ~~\tau({\rm G})=11.
\end{equation}
For any
$\bm{c}=(c_1,c_2,\cdots,c_n)\in\sum_{\rm DNA}^n$, define $\tau(\bm{c})=(\tau(c_1),\tau(c_2),\cdots,\tau(c_n))\in\mathbb{Z}_2^{2n}$.
For any two binary sequences $\bm{x}=(x_1,x_2,\cdots,x_n)$ and $\bm{y}=(y_1,y_2,\cdots,y_n)$ of length $n$, define $\bm{x}||\bm{y}=(x_1y_1,x_2y_2,\cdots,x_ny_n)\in\mathbb{Z}_2^{2n}$.

By the definition of map $\tau$ and Lemma \ref{TC-dominant}, we can directly get the following Lemma.
\begin{lemma}\label{lemB}
 Let $\bm{c}=(c_1,c_2,\cdots,c_n)\in\mathbb{Z}_2^{n}$ be ``0"-$m$-dominant, and $\bm{d}=(d_1,d_2,\cdots,d_n)\in\mathbb{Z}_2^{n}$ be $\ell$-run-length limited, then $\tau^{-1}(\bm{c}||\bm{d})$ is $m$-SSA and $\ell$-run-length limited.
\end{lemma}

Recall the work \cite{b11} presented an RLL encoder to transform any sequence in $\mathbb{Z}_q^{n-1}$ to an $\ell$-run-length limited sequence by introducing one redundant symbol on condition that $n\leq 2^{\ell-1}+\ell-1$, using the sequence replacement technique.
Next we give the following encoding algorithm.


%

Let $m, \ell, n, t$ be positive integers with odd $m\leq11$, $\ell\geq3$ and $n\leq 2^{\ell-1}+\ell-1$.

\begin{algorithm}
\caption{Construction $I$ for DNA codes }\label{ALG:1}
\begin{algorithmic}[1]
\Require $\bm{x}\in S_0(m,(n+2)t)$ and $\bm{y}\in\mathbb{Z}_2^{(n-1)t}$.
\Ensure $\bm{c}\in \sum_{\rm DNA}^{(n+2)t}$.
\State Let $\bm{y}=(\bm{y}_1,\bm{y}_2,\cdots,\bm{y}_t)$ with $|\bm{y}_i|=n-1$ for $i\in[t]$.
\State Encode each $\bm{y}_i$ to be a $\ell$-run-length limited sequence $\bm{y}'_i\in\mathbb{Z}_2^n$ by the RLL encoder in \cite{b11}. Denote $\bm{y}'_i=(y'_{i,1},y'_{i,2},\cdots,y'_{i,n})$ for $i\in[t]$.
\State For $i\in[t]$, set $z_{i,1}=\mathbb{Z}_2\setminus\{y'_{i,1}\}$ and $z_{i,2}=\mathbb{Z}_2\setminus\{y'_{i,n}\}$.
\State Update each $\bm{y}'_i$ to be $z_{i,1}\bm{y}'_iz_{i,2}$ and set $\bm{y}''=(\bm{y}'_1,\bm{y}'_2,\cdots,\bm{y}'_t)\in\mathbb{Z}_2^{(n+2)t}$.
\State Output $\bm{c}=\tau^{-1}(\bm{x}||\bm{y}'')$.
\end{algorithmic}
\end{algorithm}

One can directly obtain a decoder according to the inverse process of Algorithm~\ref{ALG:1} and the RLL decoder given in \cite{b11}. Moreover, in step $3$ of Algorithm~\ref{ALG:1}, the updated sequences $\bm{y}'_i, i\in[t]$ and $\bm{y}''$ are still $\ell$-run-length limited, since $\ell\geq3$. Thus the output codeword $\bm{c}$ is $m$-SSA and $\ell$-run-length limited by Lemma \ref{lemB}.

\begin{remark}
The DNA code $\mathcal{C}$ obtained by Algorithm~\ref{ALG:1} has length $(n+2)t$, cardinality $2^{(n-1)t}|S_0(m,(n+2)t)|$.
 It has rate $\approx1+\log_2\rho_m-3/(2^{\ell-1}+\ell+1)$ when $n=2^{\ell-1}+\ell-1$ and $t$ tends to infinity.
In particular, when $m=3$, $\ell=4$, $n=11$, the code rate tends to $\frac{10}{13}+\log_2\rho_3\approx1.3206$. Note the work \cite{BB2021} gave $3$-SSA and $4$-run-length limited DNA codes with rate 1.1609.
\end{remark}

\section{Code constructions with multiple constraints and error correction}\label{sec-4}
In this section, we firstly present a construction of 3-SSA DNA codes that is also $\ell$-run-length limited and GC-$\epsilon$-balanced for some integer $\ell\geq3$ and some real $0<\epsilon<0.5$. Then we extend it to also correct errors by using Hamming codes and quaternary VT codes.

\subsection{Codes with multiple biochemical constraints}\label{subsecA}
Recall the map $\tau$ defined in (\ref{correspondence}), we have the following Lemma.
\begin{lemma}\label{lem-1}
 Let $\bm{c}=(c_1,c_2,\cdots,c_n)\in\mathbb{Z}_2^{n}$ be ``$0$"-$3$-dominant and $\ell$-run-length limited, and $\bm{d}=(d_1,d_2,\cdots,d_n)\in\mathbb{Z}_2^{n}$ be $\epsilon$-balanced, then the DNA sequence $\tau^{-1}(\bm{c}||\bm{d})$ is $3$-SSA, $\ell$-run-length limited and GC-$\epsilon$-balanced.
\end{lemma}

It suffices to construct the two classes of binary sequences as illustrated in Lemma \ref{lem-1}.
We first construct binary sequences that are ``0"-$3$-dominant and $\ell$-run-length limited for $\ell\geq3$. Define $f(\ell, n)\subseteq\mathbb{Z}_2^n$ to be the set of binary sequences that are ``0"-$3$-dominant and $\ell$-run-length limited. Define $f_i(\ell, n)$ to be the subset of $f(\ell, n)$ where each sequence contains $i$ zeros in the first run. Then we have the following Lemma.
\begin{lemma}\label{lem-f}
The size of  $f(\ell, n)$ and $f_0(\ell, n)$ satisfy the following equations.
$$\begin{aligned}
|f(\ell, n)|=&\sum_{i=0}^\ell|f_i(\ell, n)|=\sum_{i=0}^\ell|f_0(\ell, n-i)|, \\
|f_0(\ell, n)|=&\sum_{i=0}^{\ell-2}|f_0(\ell, n-3-i)|.
\end{aligned}$$
\end{lemma}

\begin{proof}
Note that a ``0"-$3$-dominant binary sequence has no $\ell$-runs of symbol "1" since $\ell\geq3$. And $f(\ell, n)$ can be partitioned into $\ell+1$ disjoint subsets $f_i(\ell, n)$, $i\in\{0,1,\cdots,\ell\}$ according to the number of zeros in the first run of each sequence. Thus $|f(\ell, n)|=\sum_{i=0}^\ell|f_i(\ell, n)|$.
Next we build a bijection from $f_i(\ell, n)$ to $f_0(\ell, n-i)$ as follows.
$$\underbrace{0 \cdots 0}_{i}1x_{i+2}\cdots x_{n} ~~\mapsto ~~1x_{i+2}\cdots x_{n}.$$
Then $|f_i(\ell, n)|=|f_0(\ell, n-i)|$ for $i=0,...,\ell$ and $|f(\ell, n)|=\sum_{i=0}^\ell|f_i(\ell, n)|=\sum_{i=0}^\ell|f_0(\ell, n-i)|.$

On the other hand, let $\bm{x}\in f_0(\ell, n)$, then $\bm{x}$ has the form $100x_4\cdots x_n$ since $\bm{x}$ is ``0"-3-dominant. Moreover, the subsequence $x_4\cdots x_n$ following ``100" must be ``0"-$3$-dominant and $\ell$-run-length limited with no more than $\ell-2$ zeros in the first run.
Thus $|f_0(\ell, n)|=\sum_{i=0}^{\ell-2}|f_i(\ell, n-3)|=\sum_{i=0}^{\ell-2}|f_0(\ell, n-3-i)|$.
\end{proof}

By Lemma \ref{lem-f}, we have a recursive method to construct $f(\ell, n)$ once the $\ell+1$ sets $f_0(\ell, n-i)$, $i=0,\cdots,\ell$ are determined.
And $|f(\ell, n)|\approx(\ell+1)\rho^n$ when $n$ is sufficiently large, where $\rho$ is the largest real root of the equation $x^{\ell+1}-\sum_{i=0}^{\ell-2}x^{i}=0$.
To illustrate this, we give an example.

\begin{example}
Set $\ell=4$. We give $f_0(4,n)$ for $n\in\{4, 5, 6, 7\}$ in Table \ref{t1}.

{\footnotesize
\begin{table}[h]
  \centering
\caption{$f_0(4,n)$ for $n\in\{4,5,6,7\}$}\label{t1}
\begin{tabular}{c|c|c|c|c}
\hline
        $n$    &4 &5 &6 &7 \\ \hline
$|f_0(4,n)|$ &2 &3 &3 & 4 \\ \hline
$f_0(4,n)$ &\makecell[l]{1000\\ 1001} & \makecell[l]{10000, 10001 \\ 10010} & \makecell[l]{100001, 100010 \\ 100100} & \makecell[l]{1000010, 1000100 \\ 1001000, 1001001} \\ \hline
\end{tabular}
\end{table}
}
One can recursively construct $f(4, n)$ by Lemma \ref{lem-f}. We construct $f(4, n)$ for $n=7,8$ in Table \ref{t2}.

{\footnotesize
\begin{table}[h]
  \centering
\caption{$f(4,n)$ for $n=7, 8$}\label{t2}
\begin{tabular}{c|c|c}
\hline
        n   &7 &8   \\ \hline
$|f(4,n)|$ & 13&18  \\ \hline
$f(4,n)$ & \makecell[l]{0000100, 0001000\\ 0001001,
0010000 \\0010001, 0010010 \\ 0100001, 0100010\\ 0100100, 1000010 \\1000100, 1001000 \\ 1001001}
& \makecell[l]{00001000, 00001001
00010000\\ 00010001, 00010010,  00100001 \\ 00100010, 00100100,  01000010\\ 01000100, 01001000,  01001001\\
10000100, 10001000, 10001001 \\ 10010000, 10010001, 10010010
}
 \\ \hline
\end{tabular}
\end{table}
}
When $n=100$, we compute that $|f(4,100)|\approx2^{43}$.
When $n$ tends to infinity, $f_0(4, n)$ has cardinality approximately $\rho^n$ where $\rho=1.3247$ is the largest real root of the equation $x^5-x^2-x-1=0$. Then $|f(4, n)|\approx5\rho^n$.
\end{example}

Given some real $0<\epsilon<0.5$, \cite{bib13encoder} showed that there
exists a linear-time encoder ${\rm ENC}_{\epsilon}$ that encodes binary data to $\epsilon$-balanced codewords of length $n$ with only one redundant bit.
Next, we use $f(\ell,n)$ and the encoder ${\rm ENC}_{\epsilon}$ to construct a DNA code.

\begin{theorem}\label{thm1}
Let $\epsilon$ be a small real with $0<\epsilon<0.5$. Suppose $n,\ell$ are positive integers with $n>(1/\epsilon^2)\log_en$ and $\ell\geq3$.
Given the construction $f(\ell,n)$ and the $\epsilon$-balance encoder ${\rm ENC}_{\epsilon}$, define the DNA code
$$
\begin{aligned}
\mathcal{C}_{\ell,\epsilon}=\{\tau^{-1}(\bm{c}||\bm{y}): \bm{c}\in f(\ell,n),~ \bm{y}=&{\rm ENC}_{\epsilon}(\bm{x}),\\
& ~{\rm with} ~  \bm{x}\in\mathbb{Z}_2^{n-1}\}.
\end{aligned}
$$
\end{theorem}

\begin{remark}
$\mathcal{C}_{\ell,\epsilon}\subseteq\sum_{\rm DNA}^{n}$ is 3-SSA, $\ell$-run-length limited and GC-$\epsilon$-balanced.
Moreover, $\mathcal{C}_{\ell,\epsilon}$ has cardinality $2^{n-1}|f(\ell,n)|$ and rate $\frac{n-1}{n}+\frac{1}{n}\log_2|f(\ell,n)|\approx1+\log_2\rho$ when $n$ is sufficiently large, where $\rho$ is the largest real root of the equation $x^{\ell+1}-\sum_{i=0}^{\ell-2}x^{i}=0$.
Particularly, when $\ell=4, \epsilon=0.1$, $\mathcal{C}_{4, 0.1}$ has rate approximately $1+\log_2\rho\approx1.4057$.
\end{remark}

\subsection{Constrained codes with error correction}\label{subsecB}
We modify the code $\mathcal{C}_{\ell,\epsilon}$ such that it can correct errors.

{\it Construction II.
Given $\ell\geq3$ and $0<\epsilon<0.5$, the code construction includes the following two steps.
\begin{itemize}
\item {\it Extending step.} Let $\bm{y}=(y_1,y_2,\cdots, y_n)\in\mathcal{C}_{\ell,\epsilon}$, where $\mathcal{C}_{\ell,\epsilon}$ is given in Theorem \ref{thm1}. Set $\bm{z}_1(\bm{y})={\rm TC}$ if $y_1\neq{\rm C}$, otherwise set $\bm{z}_1(\bm{y})={\rm CT}$. Similarly, set $\bm{z}_2(\bm{y})={\rm TC}$ if $y_n\neq{\rm T}$, otherwise set $\bm{z}_2(\bm{y})={\rm CT}$.
    Define $S=\{\bm{z}_1(\bm{y})\bm{y}\bm{z}_2(\bm{y}): \bm{y}\in\mathcal{C}_{\ell,\epsilon}\}\subseteq\sum_{\rm DNA}^{n+4}$.
\item {\it Concatenation step.} Let $q$ be a prime power with $q\leq|S|$, and $r$ be a positive integer. Let ${\rm Ham}(r;q)$ be the $[t=\frac{q^r-1}{q-1}, t-r, 3]$ Hamming code over $\mathbb{F}_q$. Let $\pi$ be an injective map from $\mathbb{F}_q$ to the set $S$, then
    $$
    \mathcal{C}_e=\{(\pi(c_1), \pi(c_2),...,\pi(c_t)): (c_1,...,c_t)\in {\rm Ham}(r;q)\}
    $$
\end{itemize}
}

\begin{theorem} Denote $N=(n+4)t$.
The code $\mathcal{C}_e\subseteq\sum_{\rm DNA}^{N}$ in Construction II is $3$-SSA, $\ell$-run-length limited and GC-$\epsilon$-balanced for some $\ell\geq3$ and $0<\epsilon<0.5$, and can correct one substitution error. Moreover, $\mathcal{C}_e$ has {rate} $\frac{t}{N}(1-\frac{r}{t})\log_2q$. When $q$ is close to $|S|=|\mathcal{C}_{\ell, \epsilon}|$, then it has rate approximately $(1-\frac{r}{t})(1+\log_2\rho)$, where $\rho$ is the largest real root of the equation $x^{\ell+1}-\sum_{i=0}^{\ell-2}x^{i}=0$.
\end{theorem}

\begin{proof}
Let $(\pi(c_1),...,\pi(c_t))\in\mathcal{C}_e$ with some $(c_1,...,c_t)\in{\rm Ham}(r;q)$. We first prove $(\pi(c_1),...,\pi(c_t))$ is $3$-SSA, $\ell$-run-length limited and GC-$\epsilon$-balanced. By the construction, $\pi(c_i)\in S$ for $i\in[t]$.
Since every codeword $\bm{y}$ in $\mathcal{C}_{\ell,\epsilon}$ is TC-$3$-dominant, $\ell$-run-length limited and GC-$\epsilon$-balanced, it is easy to verify $\bm{z}_1(\bm{y})\bm{y}\bm{z}_2(\bm{y})\in S$ is also TC-$3$-dominant, $\ell$-run-length limited and GC-$\epsilon$-balanced by the definition of $\bm{z}_1(\bm{y}), \bm{z}_2(\bm{y})$, and the fact that $\ell\geq3$. Moreover, the concatenated codeword $(\pi(c_1),...,\pi(c_t))$ still retains these properties.
On the other hand, it is easy to see that $\mathcal{C}_e$ has Hamming distance $\geq3$, thus can correct one substitution error. The code rate can be computed since $|\mathcal{C}_e|=q^{t-r}$, where $t=\frac{q^r-1}{q-1}$.
\end{proof}


With a correspondence A $\rightarrow$ 0, T $\rightarrow$ 1, C $\rightarrow$ 2, G $\rightarrow$ 3, one can map a DNA code $\mathcal{C}$ to an isomorphic code over $\mathbb{Z}_4$, which we still call $\mathcal{C}$ if there is no ambiguity on the alphabet. Next we show the existence of a DNA code that has constraints in Section \ref{subsecA} and can also correct one edit error.

{\it Construction III. Suppose $\ell\geq3$ and $0<\epsilon<0.5$. Let $\mathcal{C}_e$ be constructed in Construction II and denote by $N=(n+4)t$ the code length of $\mathcal{C}_e$. For a pair $(a,b)\in \mathbb{Z}_{N}\times\mathbb{Z}_4$, let
\begin{eqnarray}
&&\mathcal{C}_{a,b}=\mathcal{C}_{e}\cap\mathrm{T}_{a,b}(N;4)=\{\bm{c}\in\mathcal{C}_{e}:  \nonumber  \\
&&\mathrm{Syn}(\alpha(\bm{c}))\equiv a \ (\mathrm{mod}\ N);\ \sum_{i=1}^{N}c_i\equiv b\ (\mathrm{mod}\ 4)\}, \nonumber
\end{eqnarray}
then there exists some $(a_0,b_0)\in \mathbb{Z}_{N}\times\mathbb{Z}_4$ such that $|\mathcal{C}_{a_0,b_0}|\geq\frac{|\mathcal{C}_{e}|}{4N}$ by the Pigeon Hole principle.
That is, there exists a DNA code $\mathcal{C}_{a_0,b_0}$ with code rate $R_e-\frac{2+\log_2N}{N}$, where $R_e$ is the rate of $\mathcal{C}_{e}$.
}

Note that the DNA code given in Construction III is a subcode of both $\mathcal{C}_{e}$ and a quaternary VT code, thus retains the constrained property and error correcting property of both codes.

\section{GC-$(s, \delta)$-locally balanced codes}\label{sec-5}

For an arbitrary $0<\delta<0.5$ and a positive integer $s$, we give a construction of GC-$(s, \delta)$-locally balanced code. The code also possesses 3-SSA and $\ell$-run-length limited property, and can correct one edit error.
The idea is to establish a connection between GC-locally balanced codes and GC-globally balanced codes through GC-partition balanced codes. Then we are able to construct GC-locally balanced codes using the codes constructed in previous sections.

\begin{lemma}\label{lemma-partition}
Let $0<\epsilon<0.5$ be a small real, and $s_0\in\mathbb{N}$. Suppose $\mathcal{C}$ is a GC-$\epsilon$-globally balanced code of length $s_0$, then for any $t\in\mathbb{N}$, there exists a GC-$(s_0, \epsilon)$-partition balanced code $\mathcal{C'}$ of length $ts_0$ with cardinality $|\mathcal{C'}|=|\mathcal{C}|^t$.
\end{lemma}
\begin{proof}
The proof is obvious by constructing
$
\mathcal{C'}=\{(\bm{c}_1,\bm{c}_2,...,\bm{c}_t): \bm{c}_i\in\mathcal{C}, ~ i\in[t]\}.
$
\end{proof}

Lemma \ref{lemma-partition} gives a method converting GC-$\epsilon$-globally balanced codes to GC-$(s_0, \epsilon)$-partition balanced codes.
Next we continue to give a transformation from GC-$(s_0, \epsilon)$-partition balanced codes to GC-$(s, \delta)$-locally balanced codes.

\begin{lemma}\label{lemma-local}
Let $0<\epsilon<0.5$ be a small real, and $n, s, s_0\in\mathbb{N}$ with $s_0\mid n$ and $n\geq s>2s_0$. If $\mathcal{C}$ is a GC-$(s_0, \epsilon)$-partition balanced code of length $n$, then $\mathcal{C}$ is also GC-$(s, \delta=\frac{(s_0-1)(1-2\epsilon)}{s}+\epsilon)$-locally balanced.
\end{lemma}
\begin{proof}
It is sufficient to prove that every codeword in $\mathcal{C}$ is GC-$(s, \delta)$-locally balanced.
Let $\bm{c}=(c_1,c_2,...,c_n)\in\mathcal{C}$, denote $\bm{c}=(\bm{c}^{(1)},\bm{c}^{(2)},...,\bm{c}^{(\frac{n}{s_0})})$, where $\bm{c}^{(i)}=(c_{(i-1)s_0+1},...,c_{is_0})$ for $i\in[\frac{n}{s_0}]$. We call $\bm{c}^{(i)}$ the $i$-th partition of $\bm{c}$, $\forall i\in[\frac{n}{s_0}]$.
Since $\bm{c}$ is GC-$(s_0, \epsilon)$-partition balanced, it has for each $i\in[\frac{n}{s_0}]$, $(0.5-\epsilon)s_0\leq \mathrm{wt}_{\mathrm{GC}}(\bm{c}^{(i)})\leq(0.5+\epsilon)s_0$.
Next we prove every length-$s$ substring of $\bm{c}$ is GC-$\delta$-globally balanced so that $\bm{c}$ is GC-$(s, \delta)$-locally balanced.

For any $j\in[n-s+1]$, consider the length-$s$ substring $\bm{c}_j=(c_j,c_{j+1},...,c_{j+s-1})$ of $\bm{c}$.
Since $s>2s_0$, then $\bm{c}_j$ must cover at least one complete partition $\bm{c}^{(i)}$ of $\bm{c}$ for some $i\in[\frac{n}{s_0}]$.
Suppose $\bm{c}_j$ covers in total $t\geq 1$ consecutive complete partitions $\bm{c}^{(i)},...,\bm{c}^{(i+t-1)}$, then we can write
$\bm{c}_j=(\underline{\bm{c}^{(i-1)}},\bm{c}^{(i)},...,\bm{c}^{(i+t-1)},\underline{\bm{c}^{(i+t)}})$, where $\underline{\bm{c}^{(i-1)}}$ is a substring of the partition $\bm{c}^{(i-1)}$ of length $a\in[0, s_0)$ and $\underline{\bm{c}^{(i+t)}}$ is a substring of $\bm{c}^{(i+t)}$ of length $b\in[0,s_0)$. Moreover, we have $m=a+ts_0+b$.
Then the GC-weight of $\bm{c}_j$ satisfies
\begin{equation}\label{local-GC}
(0.5-\epsilon)ts_0\leq \mathrm{wt}_{\mathrm{GC}}(\bm{c}_j)\leq(0.5+\epsilon)ts_0+a+b.
\end{equation}
Since $a,b\in[0,s_0)$, it has $a+b\leq 2(s_0-1)$ and $ts_0=s-(a+b)\geq s-2(s_0-1)$. The inequality (\ref{local-GC}) can be rewritten as follows.
\begin{equation*}
(0.5-\epsilon)(s-2(s_0-1))\leq \mathrm{wt}_{\mathrm{GC}}(\bm{c}_j)\leq(0.5+\epsilon)s+(1-2\epsilon)(s_0-1).
\end{equation*}
Then $\bm{c}_j$ is GC-$\delta$-globally balanced, if the following conditions hold:
\begin{equation*}
\begin{cases}
(0.5-\epsilon)(s-2(s_0-1))\geq(0.5-\delta)s,   \\
(0.5+\epsilon)s+(1-2\epsilon)(s_0-1)\leq(0.5+\delta)s.
\end{cases}
\end{equation*}
which gives $\delta\geq\frac{(s_0-1)(1-2\epsilon)}{s}+\epsilon$. This completes the proof.
\end{proof}

According to Lemma \ref{lemma-partition} and Lemma \ref{lemma-local}, we can give a code construction that is GC-locally balanced as follows.

{\it Construction IV. Let $\ell\geq3$ be a positive integer and $0<\epsilon<0.5$ be a small real. Suppose $s, s_0, t$ are positive integers with $ts_0\geq s>2s_0$.
Let $\mathcal{C}_{e}\subseteq\mathbb{Z}_4^{s_0}$ be the code constructed through Construction II in Section \ref{subsecB}.
Then $\mathcal{C}_{e}$ can be transformed to a GC-$(s, \delta=\frac{(s_0-1)(1-2\epsilon)}{s}+\epsilon)$-locally balanced code $\mathcal{C}_{\mathrm{local}}\subseteq\mathbb{Z}_4^{ts_0}$, using the method in Lemma \ref{lemma-partition} and Lemma \ref{lemma-local}.
Denote $N=ts_0$, and for a pair $(a,b)\in \mathbb{Z}_{N}\times\mathbb{Z}_4$, let
\begin{eqnarray}
&&\mathcal{C}_{a,b}=\mathcal{C}_{\mathrm{local}}\cap\mathrm{T}_{a,b}(N;4)=\{\bm{c}\in\mathcal{C}_{\mathrm{local}}:  \nonumber   \\
&&\mathrm{Syn}(\alpha(\bm{c}))\equiv a \ (\mathrm{mod}\ N); \ \sum_{i=1}^{N}c_i\equiv b\ (\mathrm{mod}\ 4)\},  \nonumber
\end{eqnarray}
then there exists some $(a_0,b_0)\in \mathbb{Z}_{N}\times\mathbb{Z}_4$ such that $|\mathcal{C}_{a_0,b_0}|\geq\frac{|\mathcal{C}_{\mathrm{local}}|}{4N}$ by the pigeonhole principle.
Define the DNA code to be $\mathcal{C}_{a_0,b_0}$.
}

Recall that $\mathcal{C}_{e}$ is a concatenated code with 3-SSA, $\ell$-run-length limited and GC-$\epsilon$-globally balanced property. $\mathcal{C}_{e}$ has Hamming distance at least 3. Then $\mathcal{C}_{\mathrm{local}}$ is also 3-SSA, $\ell$-run-length limited, and GC-$(s, \delta=\frac{(s_0-1)(1-2\epsilon)}{s}+\epsilon)$-locally balanced.
It is easy to verify $\mathcal{C}_{\mathrm{local}}$ also has Hamming distance at least 3. Thus the DNA code $\mathcal{C}_{a_0,b_0}$ satisfies these constraints and can correct one edit error.
Moreover, $\mathcal{C}_{a_0,b_0}$ has code rate $\frac{\log_2(|\mathcal{C}_{e}|^t/4ts_0)}{ts_0}=R_e-\frac{2+\log_2N}{N}$, where $R_e$ is the rate of $\mathcal{C}_{e}$.

\begin{remark}
Note that codes in \cite{bib12local} are GC-$\epsilon$-partition balanced and can correct single edits in each segment with a high rate $2(1-\log_2s/s)$, where the segment length $s=\Omega(\log N)$.
In our code, we allow $s$ to vary flexibly with $N$ and ensure local GC balance in each consecutive subsequence of length $s$, and also satisfy the 3-SSA and run-length limit constraint.

It's worth noting that our method in constructing $\mathcal{C}_{\ell, \epsilon}$ can adapt to binary encoders of \cite{bib13encoder, 2022ISIT} to directly enforce local GC-content constraint. However, it is not necessarily locally balanced after a straightforward concatenation.
\end{remark}


\begin{thebibliography}{10}
\bibitem{BB2021}K. G. Benerjee and A. Banerjee, ``On DNA codes with multiple constraints", IEEE Commun. Lett., vol. 25, no. 2, pp. 365-368, 2021.
\bibitem{Benerjee}K. G. Benerjee and A. Banerjee, ``On homopolymers and secondary structures avoiding, reversible, reversible-complement and GC-balanced DNA codes", IEEE Int. Symp. Inf. Theory (ISIT), 2022, pp. 204-209.
\bibitem{b10}K. Cai {\it et al.}, ``Correcting a single indel/edit for DNA-based data storage: linear-time encoders and order-optimality", IEEE Trans. Inf. Theory, vol. 67, no. 6, pp. 3438-3451, 2021.
\bibitem{bib7}H. Chu, C. Wang, Y. Zhang, ``Improved constructions of secondary structure avoidance codes for DNA sequences", arxiv:2304.11403, 2023.
\bibitem{bib12local}K. Cai, H. M. Kiah, M. Motani and T. T. Nguyen, ``Coding for segmented edits with local weight constraints", IEEE Int. Symp. Inf. Theory (ISIT), 2021, pp. 1694-1699.

\bibitem{b2}N. Goldman, P. Bertone, S.Chen {\it et al.}, ``Towards practical, high-capacity, low-maintenance information storage in synthesized DNA", Nature, vol. 494, no. 7435, pp. 77-80, Feb. 2013.
\bibitem{b12}R. Gabrys, H. M. Kiah, A. Vardy {\it et al.}, ``Locally balanced constraints," IEEE Int. Symp. Inf. Theory (ISIT), Los Angeles, CA, USA, Jun. 2020.
 \bibitem{bib9}K. A. S. Immink and K. Cai, ``Properties and constructions of constrained codes for DNA-based data storage", IEEE Access, vol. 8, pp. 49523-49531, 2020.

 \bibitem{TCS}N. Kumar, S. Siddhiprada Bhoi, A. Kumar Singh, ``A study of primer design with $w$-constacyclic shift over $\mathbb{F}_4$", TheoreticalComputer Science, vol. 960, no. C, 2023.

\bibitem{bib11capacity}Y. Liu, X. He and X. Tang, ``Capacity-achieving constrained codes with GC-content and runlength limits for DNA storage", IEEE Int. Symp. Inf. Theory (ISIT), Espoo, Finland, 2022, pp. 198-203.

\bibitem{bib1}A. Marathe, A. E. Condon, and R. M. Corn, ``On combinatorial DNA word design", J. Comput. Biol., vol. 8, no. 3, pp. 201-219, Jul. 2004.
\bibitem{b11}T. T. Nguyen, K. Cai, K. A. S. Immink and H. M. Kiah, ``Capacity-approaching constrained codes with error correction for DNA-based data storage", IEEE Trans. Inf. Theory, vol. 67, no. 8, pp. 5602-5613, 2021.
    
\bibitem{bib6}T. T. Nguyen, K. Cai, H. M. Kiah, D. T. Dao, and K. A. S. Immink, ``On the design of codes for DNA computing: secondary structure avoidance codes", arxiv:2302.13714, 2023.

\bibitem{bib13encoder}T. T. Nguyen, K. Cai, and K. A. S. Immink, ``Binary subblock energy-constrained codes: Knuth's balancing and sequence replacement techniques", IEEE Int. Symp. Inf. Theory (ISIT), Jun. 2020, pp. 37-41.
    
 

\bibitem{b9}W. Song, K. Cai, M. Zhang, and C. Yuen,``Codes with run-length and GC-content constraints for DNA-based data storage," IEEE Commun. Lett., vol. 22, no. 10, pp. 2004-2007, Oct. 2018.

\bibitem{b15}G. Tenengolts, ``Nonbinary codes, correcting single deletion or insertion", IEEE Trans. Inf. Theory, vol. 30, no. 5, pp. 766-769, Sep. 1984.

\bibitem{2022ISIT}C. Wang, Z. Lu, Z. Lan {\it et al.}, ``Coding schemes for locally balanced constraints", IEEE Int. Symp. Inf. Theory (ISIT), July, 2022.
\bibitem{bib8}Y. Wang, M. Noor-A-Rahim, E. Gunawan {\it et al.}, ``Construction of bio-constrained code for DNA data storage", IEEE Commun. Lett., vol. 23, no. 6, pp. 963-966, Jun. 2019.

 \bibitem{bib10}J. H. Weber, J. A. M. De Groot, and C. J. Van Leeuwen, ``On single-error-detecting codes for dna-based data storage", IEEE Commun. Lett., vol. 25, no. 1, pp. 41-44, Jan. 2021.



 %








\end{thebibliography}
\end{document}